\documentclass[12pt]{article}
\usepackage{amsmath,amsfonts,amssymb,amsthm,amsbsy} 

\newtheorem*{theorNo}{Theorem}

\newcommand{\cL}{\mathcal{L}}

\newcommand{\bOm}{\boldsymbol{\Omega}}
\newcommand{\bcP}{\boldsymbol{\mathcal{P}}}
\newcommand{\bQ}{\boldsymbol{Q}}
\newcommand{\bS}{\boldsymbol{S}}

\newcommand{\bs}{\boldsymbol{s}}
\newcommand{\bq}{\boldsymbol{q}}
\newcommand{\bx}{\boldsymbol{x}}
\newcommand{\bby}{\boldsymbol{y}}
\newcommand{\bz}{\boldsymbol{z}}
\newcommand{\bone}{\boldsymbol{1}}

\newcommand{\BBZ}{\mathbb{Z}}

\newcommand{\dd}{\partial}
\newcommand{\Id}{{\mathrm d}}
\newcommand{\dvol}{{\mathrm{d
vol}}}

\newcommand{\lshad}{[\![}
\newcommand{\rshad}{]\!]}

\title{The Jacobi identity for graded\/-\/commutative variational Schouten bracket revisited}
\author{A.~V.~Kiselev}

\date{\textit{E-mail: {A.V.Kiselev\symbol{"40}rug.nl}}} 

\begin{document}
\maketitle


\begin{abstract}\noindent%
This short note contains an explicit proof of the Jacobi identity for variational Schouten bracket in $\mathbb{Z}_2$-\/graded
commutative setup; an extension of the reasoning and assertion to the noncommutative geometry of cyclic words 
(see~\cite{SQS11}) is immediate, still making the proof longer. 
We emphasize that for the reasoning to be rigorous, 
it must refer to the product bundle geometry 
of iterated variations (see~\cite{gvbv}); on the other hand,
no \emph{ad hoc} regularizations occur anywhere in this theory.
\end{abstract}

\paragraph*{Introduction.}
The Jacobi identity for variational Schouten bracket $\lshad\,,\,\rshad$ is its key property in several cohomological theories. 
For example, one infers that the BV-\/Laplacian $\Delta$ or quantum BV-\/op\-e\-ra\-t\-or 
$\bOm^{\hbar}=\boldsymbol{i}\hbar\,\Delta+\lshad\bS^{\hbar},\cdot\,\rshad$
are differentials in the Batalin\/--\/Vilkovisky formalism 
(available literature is immense; let us refer 
to~\cite{gvbv} and~\cite{BV}) or one deduces that $\dd_{\bcP}=\lshad\bcP,\cdot\,\rshad$ yields the Pois\-son\/--\/Lich\-ne\-r\-o\-wicz
complex for every variational Poisson bi-\/vector~$\bcP$, 
see~\cite{SQS11}. 
Likewise, a realization of 
zero\/-\/curvature geometry for the inverse scattering via the classical 
mas\-ter\/-\/equ\-a\-ti\-on $\lshad\bS,\bS\rshad=0$
opens a way for deformation quantization, which is not restricted to the BV-\/quantization of 
Chern\/--\/Simons models over threefolds.%
\footnote{In fact, all these BV-{}, Poisson, or IST~models are examples of variational Lie algebroids~\cite{Galli10} and
their encoding by~$\bQ^2=0$. The construction of gauge automorphisms for the 
$\bQ$-\/cohomology determines the next 
generation of such structures, with new deformation quantization parameters beyond the Planck constant.}
Therefore, it is mandatory to have a clear vision of the geometry of iterated variations and understand the mechanism for
validity of the Jacobi identity.

A self\/-\/regularized calculus of variations, including the definitions 
of~$\Delta$ and $\lshad\,,\,\rshad$ and a rigorous 
proof of their interrelations, is developed in~\cite{gvbv}. 
We reserved that theory's key element, the proof of 
Theorem~4.(iii) with Jacobi's identity for $\lshad\,,\,\rshad$, 
to a separate paper which is this note. Referring to~\cite{gvbv} for detail and
discussion, let us recall that --\,in a theory of variations for fields over the space\/-\/time\,-- each integral functional%
\footnote{Let all functionals that take field configurations to number be 
\emph{integral} in this note; formal (sums of)
products of functionals such as $\exp\bigl(\tfrac{\boldsymbol{i}}{\hbar}\bS^{\hbar}\bigr)$ are dealt with by using the Leibniz rule,
see~\cite[\S\,2.5]{gvbv}.}
or every test shift of the fields brings its own copy of the domain of integration into the setup; the locality of
couplings between (co)\/vectors attached at the domains' points ensures a restriction to diagonals in the accumulated products
of bundles, whereas the operational definitions of~$\Delta$ 
and $\lshad\,,\,\rshad$ are on\/-\/the\/-\/diagonal reconfigurations
of such couplings.%
\footnote{It is readily seen from the proof of theorem below and from example
on p.~\pageref{ExJacobi} that
composite\/-\/structure objects such as brackets of functionals retain a kind of memory of the way how they were produced;
in effect, variational derivatives detect the traces of original objects' individual geometries, whence a variation within
one of them does not mar 
any of the others.}
We expect that the reader is familiar with the concept and notation 
from~\S\,1--2.4 in~\cite{gvbv}. In particular, we let
the notation for total derivatives which stem from integrations by parts keep track of the variations' arguments, so that
$\bigl((\delta s){\overleftarrow
{\dd}}\!\!
/\dd\bby\bigr)(\bby)\cdot\overrightarrow
{\dd}\!
\cL\bigl(\bx,[\bq],[\bq^{\dagger}]\bigr)/\dd\bq_{\bx}$
at $\bby=\bx$ becomes
$\delta s(\bby)\cdot\bigl(-\overrightarrow
{\Id}\!\!
/\Id\bby)(\overrightarrow
{\dd}\!
\cL\bigl(\bx,[\bq],[\bq^{\dagger}]\bigr)/\dd\bq_{\bx})$
on that diagonal, see Example on p.~\pageref{ExJacobi} 
and Example~2.4 on pp.~34--36 of~\cite{gvbv}. Similarly, the variational
derivatives with respect to (anti)\/fields~$\bq$ or~$\bq^{\dagger}$ keep track of the test shifts which those variations come
from: e.g., the formula above yields%
\footnote{In this note we let the arrow over a variational derivative indicate the direction along which all 
derivatives act --- but not the opposite direction along which the test shifts were transported prior to any
integration by parts (cf.~\cite{gvbv}); we thus have
$\overrightarrow{\delta\boldsymbol{s}}\,(\bS)=\int\!\Id\bby\,\bigl\{\bigl\langle\delta s(\bby),
\overrightarrow{\delta}\!\!/\delta\bq(\bby)\bigl(\bS(\bx)\bigr)\bigr\rangle+
\bigl\langle\delta s^{\dagger}(\bby),
\overrightarrow{\delta}\!\!/\delta\bq^{\dagger}(\bby)\bigl(\bS(\bx)\bigr)\bigr\rangle\bigr\}$ and $(\bS)\,\overleftarrow{\delta\boldsymbol{s}}
=\int\!\Id\bby\,\bigl\{\bigl\langle\bigl(\bS(\bx)\bigr)\overleftarrow{\delta}\!\!/\delta\bq(\bby),\delta s(\bby)\bigr\rangle+
\bigl\langle\bigl(\bS(\bx)\bigr)\overleftarrow{\delta}\!\!/\delta\bq^{\dagger}(\bby),\delta s^{\dagger}(\bby)\bigr\rangle\bigr\}$,
where the diagonal $\bby=\bx$ is wrought by the coupling $\langle\,,\,\rangle$, see~\cite[\S\,2.2--3]{gvbv}, and we
display 
the integration variable $\bx$ in the functional~$\bS$.%
}
a term in $\delta s(\bby)\cdot{\overrightarrow
{\delta}}\!\!
/\delta\bq(\bby)\bigl(\cL\bigl(\bx,[\bq],[\bq^{\dagger}]\bigr)\bigr)$
at~$\bby=\bx$. 
This simplifies the reasoning.%
\footnote{With a bit more care taken of the order in which the factors
follow each other in products, and by using the $\mathbb{Z}_2$-\/graded Leibniz rule for left-{} and right\/-\/directed
derivations, we show that the claim and proof of the main theorem hold true in the setup of cyclic words and brackets
of necklaces (see~\cite{SQS11} and references therein).%
}

\begin{theorNo}
Let $F$,\ $G$, and~$H$ be $\BBZ_2$-\/parity homogeneous functionals\textup{;} 
denote by $|\cdot|$ the grading so that $(-)^{|\cdot|}$
is the parity. The variational Schouten bracket $\lshad\,,\,\rshad$ satisfies the shifted\/-\/graded Jacobi identity 
\textup{(}cf.~Eq.~\textup{(28)} in Theorem~4.(iii) on p.~30 versus Eq.~\textup{(36)} on p.~37 in~\textup{\cite{gvbv}),}
\begin{equation}\label{Jacobi4Schouten}
\lshad F,\lshad G,H\rshad\rshad=\lshad\lshad F,G\rshad,H\rshad+
\smash{(-)^{(|F|-1)(|G|-1)}}\,\lshad G,\lshad F,H\rshad\rshad.
\end{equation}
The operator $\lshad F,\,\cdot\,\rshad$ is a graded derivation of 
$\lshad\,,\,\rshad$
\textup{:}\ 
identity~\eqref{Jacobi4Schouten} is the Leibniz rule for~it.
\end{theorNo}

\begin{proof}
The logic is straightforward%
\footnote{Obviously, the l.-h.s. of~\eqref{Jacobi4Schouten} does \emph{not} contain second variational derivatives of~$F$
whereas the r.-h.s.\ \emph{does}. We show that it is precisely these terms and none others which cancel out in the r.-h.s.}
as soon as the matching of (co)\/vectors and reconfigurations of couplings are understood in~\cite[\S\,1--2]{gvbv}. 
We consider first the l.-h.s.\ of~\eqref{Jacobi4Schouten}. 
By construction, we have that 
$
\lshad G,H\rshad=\bigl(G(\bx_2)\bigr)
{\overleftarrow{\delta}}\!\!/{\delta\bq(\bby_2)}\cdot
{\overrightarrow{\delta}}\!\!/{\delta\bq^\dagger(\bby_3)}\bigl(H(\bx_3)\bigr)
-\bigl(G(\bx_2)\bigr)
{\overleftarrow{\delta}}\!\!/{\delta\bq^\dagger(\bby_2)}\cdot
{\overrightarrow{\delta}}\!\!/{\delta\bq(\bby_3)}\bigl(H(\bx_3)\bigr)
$. 
Now expanding
$
\lshad F,\lshad G,H\rshad\rshad=
\bigl(F(\bx_1)\bigr)
{\overleftarrow{\delta}}\!\!/{\delta\bq(\bz_1)}\cdot
{\overrightarrow{\delta}}\!\!/{\delta\bq^{\dagger}(\bz_{23})}
\bigl(\lshad G,H\rshad\bigr)-
\bigl(F(\bx_1)\bigr)
{\overleftarrow{\delta}}\!\!/{\delta\bq^{\dagger}(\bz_1)}\cdot
{\overrightarrow{\delta}}\!\!/{\delta\bq(\bz_{23})}
\bigl(\lshad G,H\rshad\bigr)
$, 
we obtain the sum of eight enumerated terms:
\footnote{We denote by $\bz_{ij}$ the integration variables which label the variations falling --\,in the outer brackets in~\eqref{Jacobi4Schouten}\,-- 
on the $i$\/th or $j$\/th functional by the Leibniz rule 
(let $F$ be first and so on, $1\leqslant i< j\leqslant 3$); 
for convenience, we highlight $i$ in $\bz_{ij}$ when the variation falls on the $i$\/th functional --- and $j$~otherwise.}
\begin{align*}
{}^{\langle 1 \rangle}&\ %
F(\bx_1)
{\overleftarrow{\delta}}\!\!/{\delta\bq(\bz_1)}\cdot
{\overrightarrow{\delta}}\!\!/{\delta\bq^\dagger(\bz_{\boldsymbol{2}3})}\, 
G(\bx_2) 
{\overleftarrow{\delta}}\!\!/{\delta\bq(\bby_2)}\cdot
{\overrightarrow{\delta}}\!\!/{\delta\bq^\dagger(\bby_3)}\, H(\bx_3)+ {}
\\
{}{}^{\langle 2 \rangle}&\quad %
+(-)^{|G|}\,F(\bx_1)
{\overleftarrow{\delta}}\!\!/{\delta\bq(\bz_1)}\cdot
G(\bx_2) 
{\overleftarrow{\delta}}\!\!/{\delta\bq(\bby_2)}\cdot
{\overrightarrow{\delta}}\!\!/{\delta\bq^\dagger(\bz_{2\boldsymbol{3}})}
{\overrightarrow{\delta}}\!\!/{\delta\bq^\dagger(\bby_3)}\, H(\bx_3)- {}
\\
{}^{\langle 3 \rangle}&\quad %
-F(\bx_1)
{\overleftarrow{\delta}}\!\!/{\delta\bq(\bz_1)}\cdot
{\overrightarrow{\delta}}\!\!/{\delta\bq^\dagger(\bz_{\boldsymbol{2}3})}\,
\left(G(\bx_2) 
{\overleftarrow{\delta}}\!\!/{\delta\bq^\dagger(\bby_2)}\right)\cdot 
{\overrightarrow{\delta}}\!\!/{\delta\bq(\bby_3)}\, H(\bx_3)- {}
\\
{}^{\langle 4 \rangle}&\quad %
-(-)^{|G|-1}\,F(\bx_1) 
{\overleftarrow{\delta}}\!\!/{\delta\bq(\bz_1)}\cdot G(\bx_2) 
{\overleftarrow{\delta}}\!\!/{\delta\bq^\dagger(\bby_2)}\cdot
{\overrightarrow{\delta}}\!\!/{\delta\bq^\dagger(\bz_{2\boldsymbol{3}})}
{\overrightarrow{\delta}}\!\!/{\delta\bq(\bby_3)}\, H(\bx_3)- {}
\\
{}^{\langle 5 \rangle}&\quad %
-F(\bx_1)
{\overleftarrow{\delta}}\!\!/{\delta\bq^\dagger(\bz_1)}\cdot
{\overrightarrow{\delta}}\!\!/{\delta\bq(\bz_{\boldsymbol{2}3})}\,
G(\bx_2) 
{\overleftarrow{\delta}}\!\!/{\delta\bq(\bby_2)}\cdot
{\overrightarrow{\delta}}\!\!/{\delta\bq^\dagger(\bby_3)}\, H(\bx_3)- {}
\\
{}^{\langle 6 \rangle}&\quad %
-F(\bx_1)
{\overleftarrow{\delta}}\!\!/{\delta\bq^\dagger(\bz_1)}\cdot
G(\bx_2) 
{\overleftarrow{\delta}}\!\!/{\delta\bq(\bby_2)}\cdot
{\overrightarrow{\delta}}\!\!/{\delta\bq(\bz_{2\boldsymbol{3}})} 
{\overrightarrow{\delta}}\!\!/{\delta\bq^\dagger(\bby_3)}\, H(\bx_3)+ {}
\\
{}^{\langle 7 \rangle}&\quad %
+F(\bx_1)
{\overleftarrow{\delta}}\!\!/{\delta\bq^\dagger(\bz_1)}\cdot
{\overrightarrow{\delta}}\!\!/{\delta\bq(\bz_{\boldsymbol{2}3})}\,
G(\bx_2) 
{\overleftarrow{\delta}}\!\!/{\delta\bq^\dagger(\bby_2)}\cdot
{\overrightarrow{\delta}}\!\!/{\delta\bq(\bby_3)}\, H(\bx_3)+ {}
\\
{}^{\langle 8 \rangle}&\quad %
+F(\bx_1)
{\overleftarrow{\delta}}\!\!/{\delta\bq^\dagger(\bz_1)}\cdot
G(\bx_2) 
{\overleftarrow{\delta}}\!\!/{\delta\bq^\dagger(\bby_2)}\cdot
{\overrightarrow{\delta}}\!\!/{\delta\bq(\bz_{2\boldsymbol{3}})} 
{\overrightarrow{\delta}}\!\!/{\delta\bq(\bby_3)}\, H(\bx_3).
\end{align*}
Arguing as above, we see that the 
term $\lshad\lshad F,G\rshad,H\rshad$ in 
the r.-h.s.\ of~\eqref{Jacobi4Schouten} is%
\footnote{The labelling of terms by superscripts 
$\langle1\rangle$\,--\,$\langle8\rangle$ shows their matching with summands
in the l.-h.s.\ of~\eqref{Jacobi4Schouten} or, for the index running from 
$\langle9\rangle$ to $\langle12\rangle$, points at
the four 
second\/-\/order variations of $F$ which cancel out in the two r.-h.s.\ summands in Jacobi's identity.}

\begin{align*}
{}^{\langle 9 \rangle}&\ %
F(\bx_1)
{\overleftarrow{\delta}}\!\!/{\delta\bq(\bby_1)}
{\overleftarrow{\delta}}\!\!/{\delta\bq(\bz_{\boldsymbol{1}2})}\cdot
{\overrightarrow{\delta}}\!\!/{\delta\bq^\dagger(\bby_2)}\, G(\bx_2)\cdot
{\overrightarrow{\delta}}\!\!/{\delta\bq^\dagger(\bz_3)}\, H(\bx_3)+ {}
\\
{}^{\langle 1 \rangle}&\quad %
+F(\bx_1){\overleftarrow{\delta}}\!\!/{\delta\bq(\bby_1)}\cdot
{\overrightarrow{\delta}}\!\!/{\delta\bq^\dagger(\bby_2)}\, G(\bx_2)
{\overleftarrow{\delta}}\!\!/{\delta\bq(\bz_{1\boldsymbol{2}})}\cdot
{\overrightarrow{\delta}}\!\!/{\delta\bq^\dagger(\bz_3)}\, H(\bx_3)- {}
\\
{}^{\langle 10 \rangle}&\quad %
-F(\bx_1){\overleftarrow{\delta}}\!\!/{\delta\bq^\dagger(\bby_1)}
{\overleftarrow{\delta}}\!\!/{\delta\bq(\bz_{\boldsymbol{1}2})}\cdot
{\overrightarrow{\delta}}\!\!/{\delta\bq(\bby_2)}\, G(\bx_2)\cdot
{\overrightarrow{\delta}}\!\!/{\delta\bq^\dagger(\bz_3)}\, H(\bx_3)- {}
\\
{}^{\langle 5 \rangle}&\quad %
-F(\bx_1){\overleftarrow{\delta}}\!\!/{\delta\bq^\dagger(\bby_1)}\cdot
{\overrightarrow{\delta}}\!\!/{\delta\bq(\bby_2)}\, G(\bx_2)
{\overleftarrow{\delta}}\!\!/{\delta\bq(\bz_{1\boldsymbol{2}})}\cdot
{\overrightarrow{\delta}}\!\!/{\delta\bq^\dagger(\bz_3)}\, H(\bx_3)- {}
\\
{}^{\langle 11 \rangle}&\quad %
-(-)^{|G|-1}\,F(\bx_1){\overleftarrow{\delta}}\!\!/{\delta\bq(\bby_1)}
{\overleftarrow{\delta}}\!\!/{\delta\bq^\dagger(\bz_{\boldsymbol{1}2})}\cdot
{\overrightarrow{\delta}}\!\!/{\delta\bq^\dagger(\bby_2)}\, G(\bx_2)\cdot
{\overrightarrow{\delta}}\!\!/{\delta\bq(\bz_3)}\, H(\bx_3)- {}
\\
{}^{\langle 3 \rangle}&\quad %
-F(\bx_1){\overleftarrow{\delta}}\!\!/{\delta\bq(\bby_1)}\cdot
\smash{\left({\overrightarrow{\delta}}\!\!/{\delta\bq^\dagger(\bby_2)}\, G(\bx_2)\right)}
{\overleftarrow{\delta}}\!\!/{\delta\bq^\dagger(\bz_{1\boldsymbol{2}})}\cdot
{\overrightarrow{\delta}}\!\!/{\delta\bq(\bz_3)}\, H(\bx_3)+ {}
\\
{}^{\langle 12 \rangle}&\quad %
+(-)^{|G|}\,F(\bx_1){\overleftarrow{\delta}}\!\!/{\delta\bq^\dagger(\bby_1)}
{\overleftarrow{\delta}}\!\!/{\delta\bq^\dagger(\bz_{\boldsymbol{1}2})}\cdot
{\overrightarrow{\delta}}\!\!/{\delta\bq(\bby_2)}\, G(\bx_2)\cdot
{\overrightarrow{\delta}}\!\!/{\delta\bq(\bz_3)}\, H(\bx_3)+ {}
\\
{}^{\langle 7 \rangle}&\quad %
+F(\bx_1){\overleftarrow{\delta}}\!\!/{\delta\bq^\dagger(\bby_1)}\cdot
{\overrightarrow{\delta}}\!\!/{\delta\bq(\bby_2)}\, G(\bx_2)
{\overleftarrow{\delta}}\!\!/{\delta\bq^\dagger(\bz_{\boldsymbol{1}2})}\cdot
{\overrightarrow{\delta}}\!\!/{\delta\bq(\bz_3)}\, H(\bx_3).
\end{align*}
In the same way, we obtain the term $\lshad G,\lshad F,H\rshad\rshad$ not yet multiplied by the extra sign factor:
\begin{align*}
{}^{\{ 1 \} }&\ %
G(\bx_2) {\overleftarrow{\delta}}\!\!/{\delta\bq(\bz_2)} \cdot
{\overrightarrow{\delta}}\!\!/{\delta\bq^\dagger(\bz_{\boldsymbol{1}3})}
\,F(\bx_1) {\overleftarrow{\delta}}\!\!/{\delta\bq(\bby_1)} \cdot
{\overrightarrow{\delta}}\!\!/{\delta\bq^\dagger(\bby_3)}\, H(\bx_3)+ {}
\\
{}^{\{ 2 \} }&\quad %
+(-)^{|F|}\, G(\bx_2) {\overleftarrow{\delta}}\!\!/{\delta\bq(\bz_2)} \cdot
F(\bx_1) {\overleftarrow{\delta}}\!\!/{\delta\bq(\bby_1)} \cdot
{\overrightarrow{\delta}}\!\!/{\delta\bq^\dagger(\bz_{1\boldsymbol{3}})}
{\overrightarrow{\delta}}\!\!/{\delta\bq^\dagger(\bby_3)}\, H(\bx_3) - {}
\\
{}^{\{ 3 \} }&\quad %
- G(\bx_2) {\overleftarrow{\delta}}\!\!/{\delta\bq(\bz_2)} \cdot
{\overrightarrow{\delta}}\!\!/{\delta\bq^\dagger(\bz_{\boldsymbol{1}3})}
\,\smash{\left(F(\bx_1) {\overleftarrow{\delta}}\!\!/{\delta\bq^\dagger(\bby_1)}\right)}
\cdot {\overrightarrow{\delta}}\!\!/{\delta\bq(\bby_3)}\, H(\bx_3) - {}
\\
{}^{\{ 4 \} }&\quad %
-(-)^{|F|-1}\, G(\bx_2) {\overleftarrow{\delta}}\!\!/{\delta\bq(\bz_2)} \cdot
F(\bx_1) {\overleftarrow{\delta}}\!\!/{\delta\bq^\dagger(\bby_1)} \cdot
{\overrightarrow{\delta}}\!\!/{\delta\bq^\dagger(\bz_{1\boldsymbol{3}})}
{\overrightarrow{\delta}}\!\!/{\delta\bq(\bby_3)}\, H(\bx_3) - {}
\\
{}^{\{ 5 \} }&\quad %
- G(\bx_2) {\overleftarrow{\delta}}\!\!/{\delta\bq^\dagger(\bz_2)} \cdot
{\overrightarrow{\delta}}\!\!/{\delta\bq(\bz_{\boldsymbol{1}3})}
\,F(\bx_1) {\overleftarrow{\delta}}\!\!/{\delta\bq(\bby_1)} \cdot
{\overrightarrow{\delta}}\!\!/{\delta\bq^\dagger(\bby_3)}\, H(\bx_3) - {}
\\
{}^{\{ 6 \} }&\quad %
- G(\bx_2) {\overleftarrow{\delta}}\!\!/{\delta\bq^\dagger(\bz_2)} \cdot
F(\bx_1) {\overleftarrow{\delta}}\!\!/{\delta\bq(\bby_1)} \cdot
{\overrightarrow{\delta}}\!\!/{\delta\bq(\bz_{1\boldsymbol{3}})}
{\overrightarrow{\delta}}\!\!/{\delta\bq^\dagger(\bby_3)}\, H(\bx_3) + {}
\\
{}^{\{ 7 \} }&\quad %
+ G(\bx_2) {\overleftarrow{\delta}}\!\!/{\delta\bq^\dagger(\bz_2)} \cdot
{\overrightarrow{\delta}}\!\!/{\delta\bq(\bz_{\boldsymbol{1}3})}
\,F(\bx_1) {\overleftarrow{\delta}}\!\!/{\delta\bq^\dagger(\bby_1)} \cdot
{\overrightarrow{\delta}}\!\!/{\delta\bq(\bby_3)}\, H(\bx_3) + {}
\\
{}^{\{ 8 \} }&\quad %
+ G(\bx_2) {\overleftarrow{\delta}}\!\!/{\delta\bq^\dagger(\bz_2)} \cdot
F(\bx_1) {\overleftarrow{\delta}}\!\!/{\delta\bq^\dagger(\bby_1)} \cdot
{\overrightarrow{\delta}}\!\!/{\delta\bq(\bz_{1\boldsymbol{3}})}
{\overrightarrow{\delta}}\!\!/{\delta\bq(\bby_3)}\, H(\bx_3).
\end{align*}
Let us now use the $\BBZ_2$-\/graded commutativity assumption for the setup.
Transporting the variations
of $F$ leftmost, we restore the lexicographic order $F\prec G\prec H$. Finally, we multiply 
$\lshad G,\lshad F,H,\rshad\rshad$, reordered as above, by the sign factor $(-)^{(|F|-1)(|G|-1)}$; this yields%
\footnote{For each term labelled by $\{1\}$\,--\,$\{8\}$ in $\lshad G,\lshad F,H,\rshad\rshad$,\ let us calculate the product
of three signs:\ one which was written near the respective summand,\ the other which comes from the reorderings to
$F\prec G$,\ and thirdly,\ $(-)^{(|F|-1)(|G|-1)}$;\ here is the list: 
$\{1\}$: $(-)^{(|F|-1)\cdot|G|} (-)^{(|F|-1)(|G|-1)} = (-)^{|F|-1}$,\quad
$\{2\}$: $(-)^{|F|}(-)^{|F|\cdot|G|} (-)^{(|F|-1)(|G|-1)} = (-)^{|G|-1}$,\quad
$\{3\}$: $-(-)^{(|F|-2)\cdot|G|} (-)^{(|F|-1)(|G|-1)} = (-)^{|F|+|G|}$,\quad
$\{4\}$: $-(-)^{|F|-1}(-)^{(|F|-1)\cdot|G|} (-)^{(|F|-1)(|G|-1)}=-1$,\quad
$\{5\}$,\ $\{6\}$: $-(-)^{|F|\cdot(|G|-1)} (-)^{(|F|-1)(|G|-1)} = (-)^{|G|}$,\quad
$\{7\}$,\ $\{8\}$: $(-)^{(|F|-1)\cdot(|G|-1)}(-)^{(|F|-1)(|G|-1)}=+1$.
}
\begin{align*}
{}^{\langle 10 \rangle}&\ %
(-)^{|F|-1}\,
{\overrightarrow{\delta}}\!\!/{\delta\bq^\dagger(\bz_{\boldsymbol{1}3})}
\,F(\bx_1) {\overleftarrow{\delta}}\!\!/{\delta\bq(\bby_1)} \cdot
G(\bx_2) {\overleftarrow{\delta}}\!\!/{\delta\bq(\bz_2)} \cdot
{\overrightarrow{\delta}}\!\!/{\delta\bq^\dagger(\bby_3)}\, H(\bx_3) + {}
\\
{}^{\langle 2 \rangle}&\quad %
+(-)^{|G|-1}\,
F(\bx_1) {\overleftarrow{\delta}}\!\!/{\delta\bq(\bby_1)} \cdot
G(\bx_2) {\overleftarrow{\delta}}\!\!/{\delta\bq(\bz_2)} \cdot
{\overrightarrow{\delta}}\!\!/{\delta\bq^\dagger(\bz_{1\boldsymbol{3}})}
{\overrightarrow{\delta}}\!\!/{\delta\bq^\dagger(\bby_3)}\, H(\bx_3) + {}
\\
{}^{\langle 12 \rangle}&\quad %
+ (-)^{|F|+|G|}\,
{\overrightarrow{\delta}}\!\!/{\delta\bq^\dagger(\bz_{\boldsymbol{1}3})}
\,\smash{\left(F(\bx_1) {\overleftarrow{\delta}}\!\!/{\delta\bq^\dagger(\bby_1)}\right)}
\cdot G(\bx_2) {\overleftarrow{\delta}}\!\!/{\delta\bq(\bz_2)} \cdot
{\overrightarrow{\delta}}\!\!/{\delta\bq(\bby_3)}\, H(\bx_3) - {}
\\
{}^{\langle 6 \rangle}&\quad %
- F(\bx_1) {\overleftarrow{\delta}}\!\!/{\delta\bq^\dagger(\bby_1)} \cdot
G(\bx_2) {\overleftarrow{\delta}}\!\!/{\delta\bq(\bz_2)} \cdot
{\overrightarrow{\delta}}\!\!/{\delta\bq^\dagger(\bz_{1\boldsymbol{3}})}
{\overrightarrow{\delta}}\!\!/{\delta\bq(\bby_3)}\, H(\bx_3) + {}
\\
{}^{\langle 9 \rangle}&\quad %
+ (-)^{|G|}\,
{\overrightarrow{\delta}}\!\!/{\delta\bq(\bz_{\boldsymbol{1}3})}
\,F(\bx_1) {\overleftarrow{\delta}}\!\!/{\delta\bq(\bby_1)} \cdot
G(\bx_2) {\overleftarrow{\delta}}\!\!/{\delta\bq^\dagger(\bz_2)} \cdot
{\overrightarrow{\delta}}\!\!/{\delta\bq^\dagger(\bby_3)}\, H(\bx_3) + {}
\\
{}^{\langle 4 \rangle}&\quad %
+ (-)^{|G|}\,
F(\bx_1) {\overleftarrow{\delta}}\!\!/{\delta\bq(\bby_1)} \cdot
G(\bx_2) {\overleftarrow{\delta}}\!\!/{\delta\bq^\dagger(\bz_2)} \cdot
{\overrightarrow{\delta}}\!\!/{\delta\bq(\bz_{1\boldsymbol{3}})}
{\overrightarrow{\delta}}\!\!/{\delta\bq^\dagger(\bby_3)}\, H(\bx_3) + {}
\\
{}^{\langle 11 \rangle}&\quad %
+ {\overrightarrow{\delta}}\!\!/{\delta\bq(\bz_{\boldsymbol{1}3})}
\,F(\bx_1) {\overleftarrow{\delta}}\!\!/{\delta\bq^\dagger(\bby_1)} \cdot
G(\bx_2) {\overleftarrow{\delta}}\!\!/{\delta\bq^\dagger(\bz_2)} \cdot
{\overrightarrow{\delta}}\!\!/{\delta\bq(\bby_3)}\, H(\bx_3) + {}
\\
{}^{\langle 8 \rangle}&\quad %
+ F(\bx_1) {\overleftarrow{\delta}}\!\!/{\delta\bq^\dagger(\bby_1)} \cdot
G(\bx_2) {\overleftarrow{\delta}}\!\!/{\delta\bq^\dagger(\bz_2)} \cdot
{\overrightarrow{\delta}}\!\!/{\delta\bq(\bz_{1\boldsymbol{3}})}
{\overrightarrow{\delta}}\!\!/{\delta\bq(\bby_3)}\, H(\bx_3).
\end{align*}
Terms $\langle1\rangle$\,--\,$\langle8\rangle$ are present in the r.-h.s.\ of~\eqref{Jacobi4Schouten} and terms
$\langle9\rangle$\,--\,$\langle12\rangle$ cancel out;\ it is only the indices $\langle3\rangle$ and $\langle12\rangle$ which
require special attention.\ 
Consider $\langle3\rangle$ in $\lshad\lshad F,G\rshad,H\rshad$;\ %
by re\-la\-bel\-ling 
the integration variables, $\bby\rightleftarrows\bz$ (i.e.,\ 
swapping the test shifts,\ see p.\,\pageref{ExJacobi}),\ we obtain
\begin{equation}\label{EqTerm3}
-F(\bx_1)\smash{ {\overleftarrow{\delta}}\!\!/{\delta\bq(\bz_1)} }\cdot
\smash{ \left({\overrightarrow{\delta}}\!\!/{\delta\bq^\dagger(\bz_{1\boldsymbol{2}})}\, G(\bx_2)\right) }
\smash{ {\overleftarrow{\delta}}\!\!/{\delta\bq^\dagger(\bby_2)} }\cdot
\smash{ {\overrightarrow{\delta}}\!\!/{\delta\bq(\bby_3)} }\, H(\bx_3).
\end{equation}
The variation's argument in parentheses has grading $|G|-1$, which yields the sign factor $(-)^{(|G|-1)-1}$ when the left\/-\/acting
parity\/-\/odd variation $\overleftarrow{\delta}\!\!/\delta\bq^{\dagger}(\bby_2)$ is brought to the other side of its argument,
becoming $\overrightarrow{\delta}\!\!/\delta\bq^{\dagger}(\bby_2)$. 
Hence 
$(-)^{|G|-2}
{\overrightarrow{\delta}}\!\!/{\delta\bq^{\dagger}(\bby_2)}
\bigl(
{\overrightarrow{\delta}}\!\!/{\delta\bq^{\dagger}(\bz_{\boldsymbol{2}3})}
\,\bigl(G(\bx_2)\bigr)
\bigr)
\stackrel{(i)}{=}
(-)^{|G|-1}
{\overrightarrow{\delta}}\!\!/{\delta\bq^{\dagger}(\bz_{\boldsymbol{2}3})}
\bigl(
{\overrightarrow{\delta}}\!\!/{\delta\bq^{\dagger}(\bby_{2})}\,
\bigl(G(\bx_2)\bigr)
\bigr)
\stackrel{(ii)}{=}
(-)^{|G|-1}(-)^{|G|-1}
{\overrightarrow{\delta}}\!\!/{\delta\bq^{\dagger}(\bz_{\boldsymbol{2}3})}
\bigl(
\bigl(G(\bx_2)\bigr)
{\overleftarrow{\delta}}\!\!/{\delta\bq^{\dagger}(\bby_{2})}
\bigr)$,
where $(i)$ the parity\/-\/odd variations are swapped and $(ii)$ the inner variational derivative is transported around $G$ of grading $|G|$. 
The two sign factors cancel out, and the overall minus matches that near $\langle3\rangle$ in the l.-h.s.\ of~\eqref{Jacobi4Schouten}.

We do the same with $\langle12\rangle$. Consider such term in
$(-)^{(|F|-1)(|G|-1)}\lshad G,\lshad F,H\rshad\rshad$;
clearly, the factor $(-)^{|G|}$ is irrelevant because it is present also near $\langle12\rangle$ in
$\lshad\lshad F,G\rshad,H\rshad$. Transporting the parity\/-\/odd variation 
$\smash{ \overrightarrow{\delta}\!\!/\delta\bq^{\dagger}(\bz_{\bone3}) }$ around the object of grading $|F|-1$ in parentheses, we gain the
factor $(-)^{|F|-2}$, which cancels out with $(-)^{|F|}$. 
Next, relabel $\bby\rightleftarrows\bz$, which gives
\begin{equation}\label{EqTerm12}
F(\bx_1) \smash{ {\overleftarrow{\delta}}\!\!/{\delta\bq^\dagger(\bz_{\boldsymbol{1}3})} }
\smash{ {\overleftarrow{\delta}}\!\!/{\delta\bq^\dagger(\bby_1)} }
\cdot G(\bx_2) \smash{ {\overleftarrow{\delta}}\!\!/{\delta\bq(\bby_2)} }\cdot
\smash{ {\overrightarrow{\delta}}\!\!/{\delta\bq(\bz_{1\boldsymbol{3}})} }\, H(\bx_3).
\end{equation}
The parity\/-\/odd variations follow in the order which is reverse with respect to that in $\langle12\rangle$ in 
$\lshad\lshad F,G\rshad,H\rshad$, hence these terms cancel out. The proof is complete.
\end{proof}

\paragraph*{Conclusion.}
Variations $\delta\bs$ act via graded Leibniz rule on products of integral functionals,\ e.g.,\ $F\cdot\lshad G,H\rshad$;\ within 
composite objects
like $\lshad G,H\rshad$,\ they act also by derivation w.r.t.\ own geometries of
the blocks $G,H$; variations are graded\/-\/permutable in each block.\ %
Neither $\Delta$ nor $\lshad\,,\,\rshad$ depend on a choice of normalized test shift $\delta\bs$. This yields~\eqref{Jacobi4Schouten} and $\Delta^2(F\cdot G\cdot H)=0$.

\paragraph*{Acknowledgements.}
The author thanks the Organizing committee
of International workshop SQS'13 
`Supersymmetry and Quantum Symmetries'
(July~29 -- August~3, 2013; JINR Dubna, Russia)
for stimulating 
discussions and partial financial support. 

This research was supported in part by JBI~RUG project~103511 (Groningen). 
A~part of this research was done while the author was visiting at 
the $\smash{\text{IH\'ES}}$ (Bures\/-\/sur\/-\/Yvette); 
the financial support and hospitality of this institution are gratefully acknowledged.

\newpage
\setcounter{page}{1}
\renewcommand{\thepage}{\roman{page}}
\paragraph*{Appendix: an example.}\label{ExJacobi}
Let us illustrate the validity mechanism for Jacobi identity~\eqref{Jacobi4Schouten} by verifying it at three given
functionals. For simplicity, let there be just one independent variable~$x$, 
one parity\/-\/even coordinate~$q$ and its 
parity\/-\/odd canonically conjugate~$q^{\dagger}$. Set 
\[
F=\int q^{\dagger}qq_{x_1x_1}\,\dvol(x_1),\quad
G=\int q^{\dagger}_{x_2}\exp(q_{x_2})\,\dvol(x_2),\quad
\text{and } H=\int q^{\dagger}_{x_3x_3}\cos q\,\dvol(x_3);
\]
we note that the functionals $F$ and~$H$ re\/-\/appear in~\cite[pp.\,34--36]{gvbv}. We have $|F|=1$ and $|G|=1$, whence 
$(-)^{(|F|-1)(|G|-1)}=+1$ in~\eqref{Jacobi4Schouten}.
 
Let $\delta\bs_1=(\delta s_1,\delta s_1^{\dagger})$ and 
$\delta\bs_2=(\delta s_2,\delta s_2^{\dagger})$ be two normalized test shifts, i.\,e., suppose that 
$\delta s_{\alpha}(y)\cdot\delta s_{\alpha}^{\dagger}(y)=1$
at every~$y$ for $\alpha=1,2$. We recall from Lemma~1 in~\cite[p.~24]{gvbv} 
that the values of Schouten brackets
in~\eqref{Jacobi4Schouten} are independent of a concrete choice of the normalized functional coefficients 
$\delta s_{\alpha}$ and $\delta s^{\dagger}_{\alpha}$, which implies that the test shifts $\delta\bs_1$ and $\delta\bs_2$
in the inner and outer brackets can be swapped (this would amount to relabelling $y\rightleftarrows z$ of their arguments,
see~\eqref{EqTerm3} and~\eqref{EqTerm12} on p.~\pageref{EqTerm3}).
In the same way as we did in Example~2.4 in~\cite{gvbv}, we now do not write the basic (co)\/vectors $\vec{e}(y)$ and
$\vec{e}^{\mathstrut\,\,\dagger}(y)$ in expansions of the test shifts and differentials of densities of the functionals 
(see~\cite[\S\,2.2--3]{gvbv} for detail); it is enough to know the 
couplings' values, which are~$\pm1$.

We have that%
\footnote{Let us recall that integrations by parts, which cast the derivatives off the test shifts, are performed only 
when all the objects --\,such as the l.-h.s.\ or r.-h.s.\ of~\eqref{Jacobi4Schouten}\,-- are fully composed, all partial
derivatives of the functionals' densities are calculated, and reconfigurations of the couplings are ready to start.
In practice, this means that partial derivatives like $\overrightarrow{\dd}\!\!/\dd q_x$ or 
$\overleftarrow{\dd}\!\!/\dd q_{xx}^{\dagger}$ dive under $\overrightarrow{\Id}\!\!/\Id y$ or 
$\overleftarrow{\Id}\!\!/\Id z$ 
because those total derivatives 
have not yet appeared at the places where we write them ahead of time.}
$\lshad G,H\rshad={}$
\begin{multline*}
\smash{\iiiint} \Id y_2\,\Id y_3\,\Id x_2\,\dvol(x_3)\cdot
\Bigl\{\Bigl\langle\bigl(-\tfrac{\Id}{\Id y_2}\bigr)
\bigl(\underbrace{ q^\dagger_{x_2}\,\exp(q_{x_2}) }_{x_2}\bigr)\cdot
\tfrac{\Id^2}{\Id y_3^2}\,(\underbrace{ \cos q }_{x_3})\Bigr\rangle \cdot
\underbrace{\langle \delta s(y_2),\delta s^\dagger(y_3) \rangle}_{+1}
+{}\\
+\Bigl\langle \bigl(-\tfrac{\Id}{\Id y_2}\bigr)
{ \bigl(\underbrace{\exp(q_{x_2})}_{x_2}\bigr)\cdot\underbrace{q^\dagger_{x_3x_3}\cdot(-\sin q)}_{x_3} }\Bigr\rangle
\cdot \smash{ \underbrace{\langle \delta s^\dagger(y_2),\delta s(y_3) \rangle}_{-1} }
\Bigr\};
\end{multline*}
as usual, we display the integration variables~$x_i$ under the remnants of respective densities. Next, we obtain that
$\lshad F,\lshad G,H\rshad\rshad={}$
\begin{align*}
{}&\smash{ \int\!\Id z_1\int\!\Id z_{23}
\int\!\Id y_2\int\!\Id y_3\int\!\Id x_1\int\!\Id x_2\int\!\dvol(x_3) }\cdot
\underbrace{\langle \delta s(z_1),\delta s^\dagger(z_{23})\rangle}_{+1}\cdot{}
\\
{}&\ \Bigl\{
\Bigl\langle\bigl(\underbrace{{}^{\langle 1\rangle}\ %
q^\dagger q_{x_1x_1} + {}^{\langle 2\rangle}\ %
\tfrac{\Id^2}{\Id z_1^2}(q^\dagger q)}_{x_1}\bigr) \cdot
\bigl(-\tfrac{\Id}{\Id z_{\boldsymbol{2}3}}\bigr)
\bigl(-\tfrac{\Id}{\Id y_2}\bigr)\bigl(\underbrace{\exp(q_{x_2})}_{x_2}\bigr)
\cdot \tfrac{\Id^2}{\Id y_3^2} (\underbrace{\cos q}_{x_3})\Bigr\rangle \cdot
\underbrace{\langle \delta s(y_2),\delta s^\dagger(y_3)\rangle}_{+1} +{}
\\
{}&{}\quad
+\Bigl\langle\bigl(\underbrace{{}^{\langle 3\rangle}\ %
q^\dagger q_{x_1x_1} + {}^{\langle 4\rangle}\ %
\tfrac{\Id^2}{\Id z_1^2}(q^\dagger q)}_{x_1}\bigr) \cdot
\bigl(-\tfrac{\Id}{\Id y_2}\bigr)\bigl(\underbrace{\exp(q_{x_2})}_{x_2}\bigr)
\cdot 
\tfrac{\Id^2}{\Id z_{2\boldsymbol{3}}^2 } 
(\underbrace{-\sin q}_{x_3})\Bigr\rangle \cdot
\underbrace{\langle \delta s^\dagger(y_2),\delta s(y_3)\rangle}_{-1} 
\Bigr\} +{}
\\
{}&{}+\int\!\Id z_1\int\!\Id z_{23}
\int\!\Id y_2\int\!\Id y_3\int\!\Id x_1\int\!\Id x_2\int\!\dvol(x_3)\cdot
\underbrace{\langle \delta s^\dagger(z_1),\delta s(z_{23})\rangle}_{-1}\cdot{}
\\
{}&\ \Bigl\{\Bigl\langle
\bigl(\underbrace{{}^{\langle 5\rangle}\ q q_{x_1x_1} }_{x_1}\bigr) \cdot
\bigl(-\tfrac{\Id}{\Id z_{\boldsymbol{2}3}}\bigr)
\bigl(-\tfrac{\Id}{\Id y_2}\bigr)\bigl(\underbrace{q^\dagger_{x_2}\,\exp(q_{x_2})}_{x_2}\bigr)
\cdot \tfrac{\Id^2}{\Id y_3^2} (\underbrace{\cos q}_{x_3}) +{}
\\
{}&{}\qquad\qquad 
+ \bigl(\underbrace{{}^{\langle 6\rangle}\ q q_{x_1x_1} }_{x_1}\bigr) \cdot
\bigl(-\tfrac{\Id}{\Id y_2}\bigr)\bigl(\underbrace{q^\dagger_{x_2}\,\exp(q_{x_2})}_{x_2}\bigr)
\cdot 
\tfrac{\Id^2}{\Id y_3^2 } 
(\underbrace{-\sin q}_{x_3})
\Bigr\rangle\cdot
\underbrace{\langle \delta s(y_2),\delta s^\dagger(y_3)\rangle}_{+1} +{}
\end{align*}
\begin{align*}
{}&{}\quad
+\Bigl\langle
\bigl(\underbrace{{}^{\langle 7\rangle}\ q q_{x_1x_1} }_{x_1}\bigr) \cdot
\bigl(-\tfrac{\Id}{\Id z_{\boldsymbol{2}3} }\bigr)
\bigl(-\tfrac{\Id}{\Id y_2}\bigr)\bigl(\underbrace{\exp(q_{x_2})}_{x_2}\bigr)
\cdot (\underbrace{ q^\dagger_{x_3x_3}\cdot (-\sin q)}_{x_3}) +{}
\\
{}&{}\qquad\qquad 
{}+\bigl(\underbrace{{}^{\langle 8\rangle}\ q q_{x_1x_1} }_{x_1}\bigr) \cdot
\bigl(-\tfrac{\Id}{\Id y_2}\bigr) \bigl(\underbrace{\exp(q_{x_2})}_{x_2}\bigr)
\cdot (\underbrace{ q^\dagger_{x_3x_3}\cdot (-\cos q)}_{x_3})
\Bigr\rangle\cdot
\underbrace{\langle \delta s^\dagger(y_2),\delta s(y_3)\rangle}_{-1} \Bigr\}.
\end{align*}
On the other hand, $\lshad F,G\rshad={}$
\begin{multline*}
\iiiint \Id y_1\,\Id y_2\,\Id x_1\,\dvol(x_2)\cdot
\Bigl\{
\Bigl\langle
\bigl(\underbrace{q^\dagger q_{x_1x_1} + 
\tfrac{\Id^2}{\Id y_1^2}(q^\dagger q)}_{x_1}\bigr) \cdot
\bigl(-\tfrac{\Id}{\Id y_2}\bigr)\bigl(\underbrace{\exp(q_{x_2})}_{x_2}\bigr)
\Bigr\rangle \cdot
\underbrace{\langle \delta s(y_1),\delta s^\dagger(y_2) \rangle}_{+1}
+{}\\
+\Bigl\langle \bigl(\underbrace{ q q_{x_1x_1} }_{x_1}\bigr) \cdot
\bigl(-\tfrac{\Id}{\Id y_2}\bigr)\bigl(\underbrace{q^\dagger_{x_2}\,\exp(q_{x_2})}_{x_2}\bigr) \Bigr\rangle \cdot 
\underbrace{\langle \delta s^\dagger(y_1),\delta s(y_2) \rangle}_{-1}
\Bigr\}.
\end{multline*}
We infer that $\lshad\lshad F,G\rshad,H\rshad={}$
\begin{align*}
{}&\int\!\Id z_{12}\int\!\Id z_3
\int\!\Id y_1\int\!\Id y_2\int\!\Id x_1\int\!\Id x_2\int\!\dvol(x_3)\cdot
\underbrace{\langle \delta s(z_{12}),\delta s^\dagger(z_3)\rangle}_{+1}\cdot{}
\\
{}&\ \Bigl\{
\Bigl\langle \bigl( \underbrace{ {}^{\langle 9\rangle}\ %
\tfrac{\Id^2}{\Id z^2_{\boldsymbol{1}2}} (q^\dagger)+ 
{}^{\langle 10\rangle}\ %
\tfrac{\Id^2}{\Id y_1^2} (q^\dagger) }_{x_1} \bigr) \cdot
\bigl(-\tfrac{\Id}{\Id y_2}\bigr) \bigl(\underbrace{\exp(q_{x_2})}_{x_2}\bigr)
\cdot \tfrac{\Id^2}{\Id z_3^2 } (\underbrace{\cos q}_{x_3}) +{}
\\
{}&{}\qquad\quad 
+ \bigl(\underbrace{{}^{\langle 1\rangle}\ %
q^\dagger q_{x_1x_1} + {}^{\langle 2\rangle}\ %
\tfrac{\Id^2}{\Id y_1^2}(q^\dagger q)}_{x_1}\bigr) \cdot
\bigl(-\tfrac{\Id}{\Id z_{1\boldsymbol{2}}}\bigr)
\bigl(-\tfrac{\Id}{\Id y_2}\bigr)\bigl(\underbrace{\exp(q_{x_2})}_{x_2}\bigr)
\cdot \tfrac{\Id^2}{\Id z_3^2} (\underbrace{\cos q}_{x_3})
\Bigr\rangle \cdot
\underbrace{\langle \delta s(y_1),\delta s^\dagger(y_2)\rangle}_{+1} +{}
\\
{}&{}\quad
+\Bigl\langle
\bigl(\underbrace{{}^{\langle 11\rangle}\ %
q_{x_1x_1} + {}^{\langle 12\rangle}\ %
\tfrac{\Id^2}{\Id z_{\boldsymbol{1}2}^2}(q)}_{x_1}\bigr) \cdot
\bigl(-\tfrac{\Id}{\Id y_2}\bigr)\bigl(\underbrace{q^\dagger_{x_2}\,\exp(q_{x_2})}_{x_2}\bigr)
\cdot \tfrac{\Id^2}{\Id z_3^2} (\underbrace{\cos q}_{x_3}) + {}
\\
{}&{}\qquad\quad 
+ \bigl(\underbrace{{}^{\langle 5\rangle}\ q q_{x_1x_1} }_{x_1}\bigr) \cdot
\bigl(-\tfrac{\Id}{\Id z_{1\boldsymbol{2}}}\bigr)
\bigl(-\tfrac{\Id}{\Id y_2}\bigr)\bigl(\underbrace{q^\dagger_{x_2}\,\exp(q_{x_2})}_{x_2}\bigr)
\cdot \tfrac{\Id^2}{\Id z_3^2} (\underbrace{\cos q}_{x_3})
\Bigr\rangle \cdot
\underbrace{\langle \delta s^\dagger(y_1),\delta s(y_2)\rangle}_{-1}
\Bigr\} + {}
\\
{}&{}+\int\!\Id z_{12}\int\!\Id z_3
\int\!\Id y_1\int\!\Id y_2\int\!\Id x_1\int\!\Id x_2\int\!\dvol(x_3)\cdot
\underbrace{\langle \delta s^\dagger(z_{12}),\delta s(z_3)\rangle}_{-1}\cdot{}
\\
{}&\ \Bigl\{
\Bigl\langle
\bigl(\underbrace{{}^{\langle 13\rangle}\ %
q_{x_1x_1} + {}^{\langle 14\rangle}\ %
\tfrac{\Id^2}{\Id y_{1}^2}(q)}_{x_1}\bigr) \cdot
\bigl(-\tfrac{\Id}{\Id y_2}\bigr)\bigl(\underbrace{\exp(q_{x_2})}_{x_2}\bigr)
\cdot (\underbrace{ q^\dagger_{x_3x_3}\cdot (-\sin q)}_{x_3})
\Bigr\rangle \cdot
\underbrace{\langle \delta s(y_1),\delta s^\dagger(y_2)\rangle}_{+1} +{}
\\
{}&{}\quad
+\Bigl\langle
\bigl(\underbrace{{}^{\langle 7\rangle}\ q q_{x_1x_1} }_{x_1}\bigr) \cdot
\bigl(-\tfrac{\Id}{\Id z_{1\boldsymbol{2}} }\bigr)
\bigl(-\tfrac{\Id}{\Id y_2}\bigr)\bigl(\underbrace{\exp(q_{x_2})}_{x_2}\bigr)
\cdot (\underbrace{ q^\dagger_{x_3x_3}\cdot (-\sin q)}_{x_3})
\Bigr\rangle\cdot
\underbrace{\langle \delta s^\dagger(y_1),\delta s(y_2)\rangle}_{-1}
\Bigr\}.
\end{align*}
Thirdly, $\lshad F,H\rshad={}$
\begin{multline*}
\iiiint \Id y_1\,\Id y_3\,\Id x_1\,\dvol(x_3)\cdot
\Bigl\{
\Bigl\langle
\bigl(\underbrace{ q^\dagger q_{x_1x_1} + 
\tfrac{\Id^2}{\Id y_1^2}(q^\dagger q)}_{x_1}\bigr) \cdot
\tfrac{\Id^2}{\Id y_3^2} (\underbrace{\cos q}_{x_3})
\Bigr\rangle \cdot
\underbrace{\langle \delta s(y_1),\delta s^\dagger(y_3) \rangle}_{+1}
+{}\\
+\Bigl\langle 
\bigl(\underbrace{ q q_{x_1x_1} }_{x_1}\bigr) \cdot
(\underbrace{ q^\dagger_{x_3x_3}\cdot (-\sin q)}_{x_3})
\Bigr\rangle \cdot 
\underbrace{\langle \delta s^\dagger(y_1),\delta s(y_3) \rangle}_{-1}
\Bigr\}.
\end{multline*}
In view of the functionals' gradings, we have
$+1\cdot\lshad G,\lshad F,H\rshad\rshad={}$
\begin{align*}
{}&\int\!\Id z_2\int\!\Id z_{13}
\int\!\Id y_1\int\!\Id y_3\int\!\Id x_1\int\!\Id x_2\int\!\dvol(x_3)\cdot
\underbrace{\langle \delta s(z_2),\delta s^\dagger(z_{13})\rangle}_{+1}\cdot{}
\\
{}&\ \Bigl\{\Bigl\langle
\bigl(-\tfrac{\Id}{\Id z_2}\bigr)\bigl(\underbrace{q^\dagger_{x_2}\,\exp(q_{x_2})}_{x_2}\bigr)
\cdot \bigl(\underbrace{{}^{\langle 11\rangle}\ %
q_{x_1x_1} + {}^{\langle 12\rangle}\ %
\tfrac{\Id^2}{\Id y_1^2}(q)}_{x_1}\bigr) \cdot
\tfrac{\Id^2}{\Id y_3^2} (\underbrace{\cos q}_{x_3})
\Bigr\rangle \cdot
\underbrace{\langle \delta s(y_1),\delta s^\dagger(y_3) \rangle}_{+1} +{}
\\
{}&{}\qquad\quad 
+\Bigl\langle
\bigl(-\tfrac{\Id}{\Id z_2}\bigr)\bigl(\underbrace{q^\dagger_{x_2}\,\exp(q_{x_2})}_{x_2}\bigr) \cdot 
\bigl(\underbrace{{}^{\langle 6\rangle}\ q q_{x_1x_1} }_{x_1}\bigr) \cdot
\tfrac{\Id^2}{\Id z_{1\boldsymbol{3}}^2 } (\underbrace{-\sin q}_{x_3})
\Bigr\rangle \cdot 
\underbrace{\langle \delta s^\dagger(y_1),\delta s(y_3) \rangle}_{-1}
\Bigr\} +{}
\\
{}&{}+\int\!\Id z_2\int\!\Id z_{13}
\int\!\Id y_1\int\!\Id y_3\int\!\Id x_1\int\!\Id x_2\int\!\dvol(x_3)\cdot
\underbrace{\langle \delta s^\dagger(z_2),\delta s(z_{13})\rangle}_{-1}\cdot{}
\\
{}&\ \Bigl\{\Bigl\langle
\bigl(-\tfrac{\Id}{\Id z_2}\bigr) \bigl(\underbrace{\exp(q_{x_2})}_{x_2}\bigr)
\cdot \bigl( \underbrace{ {}^{\langle 10\rangle}\ %
\tfrac{\Id^2}{\Id z^2_{\boldsymbol{1}3}} (q^\dagger)+ 
{}^{\langle 9\rangle}\ %
\tfrac{\Id^2}{\Id y_1^2} (q^\dagger) }_{x_1} \bigr) \cdot
\tfrac{\Id^2}{\Id y_3^2 } (\underbrace{\cos q}_{x_3}) +{}
\\
{}&{}\qquad\quad 
+ \bigl(-\tfrac{\Id}{\Id z_2}\bigr)\bigl(\underbrace{\exp(q_{x_2})}_{x_2}\bigr)
\cdot \bigl(\underbrace{{}^{\langle 3\rangle}\ %
q^\dagger q_{x_1x_1} + {}^{\langle 4\rangle}\ %
\tfrac{\Id^2}{\Id y_1^2}(q^\dagger q)}_{x_1}\bigr) \cdot
\tfrac{\Id^2}{\Id y_3^2 } (\underbrace{-\sin q}_{x_3})
\Bigr\rangle \cdot
\underbrace{\langle \delta s(y_1),\delta s^\dagger(y_3) \rangle}_{+1} +{}
\\
{}&{}\quad
+\Bigl\langle
\bigl(-\tfrac{\Id}{\Id z_2}\bigr)\bigl(\underbrace{\exp(q_{x_2})}_{x_2}\bigr)
\cdot \bigl(\underbrace{{}^{\langle 13\rangle}\ %
q_{x_1x_1} + {}^{\langle 14\rangle}\ %
\tfrac{\Id^2}{\Id z_{\boldsymbol{1}3}^2}(q)}_{x_1}\bigr) \cdot
(\underbrace{ q^\dagger_{x_3x_3}\cdot (-\sin q)}_{x_3}) + {}
\\
{}&{}\qquad\quad 
+ \bigl(-\tfrac{\Id}{\Id z_2}\bigr) \bigl(\underbrace{\exp(q_{x_2})}_{x_2}\bigr)
\cdot \bigl(\underbrace{{}^{\langle 8\rangle}\ q q_{x_1x_1} }_{x_1}\bigr) \cdot
(\underbrace{ q^\dagger_{x_3x_3}\cdot (-\cos q)}_{x_3})
\Bigr\rangle \cdot 
\underbrace{\langle \delta s^\dagger(y_1),\delta s(y_3) \rangle}_{-1}
\Bigr\}.
\end{align*}
Each term $\langle1\rangle$\,--\,$\langle8\rangle$ meets its match in the other side of~\eqref{Jacobi4Schouten}, whereas
terms $\langle9\rangle$\,--\,$\langle14\rangle$ occur in pairs of opposite signs; therefore, they all cancel out in the
r.-h.s.\ of the Jacobi identity.\\
\centerline{\rule{1in}{0.7pt}}

\begin{thebibliography}{7}

\bibitem{SQS11}
{Kiselev A. V.} On the variational noncommutative Poisson geometry //
{Physics of 
Particles and 
Nuclei}. 
2012. V.~{43}, n.5. P.~663--665. 

\bibitem{gvbv}
{Kiselev A. V.} 
The geometry of variations in Batalin\/--\/Vilkovisky formalism //
{J.~Phys.:\ Conf.\ Ser.} 2013. V.~{474}, n.012024. P.~1--51.

\bibitem{BV}
{Batalin I., Vilkovisky G.} 
Gauge algebra and quantization // {Phys.\ Lett.}~B.
1981. V.~102, n.1. P.~27--31;\quad
{Batalin I. A., Vilkovisky G. A.}
Quantization of gauge theories with linearly dependent generators //
{Phys.\ Rev.}~D. 1983. V.~29, n.10. P.~2567--2582. 




\bibitem{Galli10}
{Kiselev A. V., van de Leur J. W.} 
Variational Lie algebroids and homological evolutionary vector fields //
{Theor.\ Math.\ Phys.} 2011. V.~167, n.3. P.~772--784.

\end{thebibliography}
\end{document}